\documentclass{llncs}

\usepackage[left=2.5cm,right=2.5cm,top=2.5cm,bottom=2.5cm]{geometry}

\usepackage{makeidx}
\usepackage{graphicx,latexsym,alltt}
\usepackage{amssymb,amsmath}
\usepackage{algorithm,algorithmic}
\usepackage{subfigure}
\usepackage{color}
\usepackage{url}
\usepackage{proof}
\usepackage{multirow}

\newcommand{\figures}[1]{./Figures/#1}
\newcommand{\biblio}[1]{./#1}

\definecolor{blank}{rgb}{0.7,0.7,0.7}

\long\def\comment#1{}

\renewcommand{\phi}{\varphi}

\def\defemb#1#2{\expandafter\def\csname #1\endcsname
                              {\relax\ifmmode #2\else\hbox{$#2$}\fi}}
\defemb{cA}{{\cal A}}
\defemb{cB}{{\cal B}}
\defemb{cC}{{\cal C}}
\defemb{cD}{{\cal D}}
\defemb{cE}{{\cal E}}
\defemb{cF}{{\cal F}}
\defemb{cG}{{\cal G}}
\defemb{cH}{{\cal H}}
\defemb{cI}{{\cal I}}
\defemb{cJ}{{\cal J}}
\defemb{cL}{{\cal L}}
\defemb{cM}{{\cal M}}
\defemb{cN}{{\cal N}}
\defemb{cO}{{\cal O}}
\defemb{cP}{{\cal P}}
\defemb{cR}{{\cal R}}
\defemb{cS}{{\cal S}}
\defemb{cT}{{\cal T}}
\defemb{cU}{{\cal U}}
\defemb{cV}{{\cal V}}
\defemb{cW}{{\cal W}}
\defemb{cX}{{\cal X}}
\defemb{cZ}{{\cal Z}}

\makeatletter
\newenvironment{prog}{\vspace{1.0ex}\par
\obeylines\@vobeyspaces\tt}{\vspace{1.0ex}\noindent
}
\makeatother
\newcommand{\startprog}{\begin{prog}}
\newcommand{\stopprog}{\end{prog}\noindent}

\catcode`\_=\active
\let_=\sb
\catcode`\_=12
\makeatother

\begin{document}

\frontmatter
\pagestyle{headings}
\addtocmark{Hamiltonian Mechanics}

\title
{
 Transforming \emph{while/do/for/foreach}-Loops\\ into Recursive Methods%
 \thanks
 {
This work has been partially supported by the EU (FEDER) and the
Spanish \emph{Ministerio de Econom\'{\i}a y Competitividad
(Secretar\'{\i}a de Estado de Investigaci\'on, Desarrollo e Innovaci\'on)}
under grant TIN2013-44742-C4-1-R and by the
\emph{Generalitat Valenciana} under grant PROMETEO/2011/052.
David Insa was partially supported
by the Spanish {Ministerio de Educaci\'on} under FPU grant AP2010-4415.
 }
}
\titlerunning{Transforming \emph{while/do/for/foreach}-Loops into Recursive Methods}

\author{David Insa \and Josep Silva}
\institute
{
Departamento de Sistemas Inform\'aticos y Computaci\'on\\
Universitat Polit\`ecnica de Val\`encia\\
Camino de Vera s/n\\
E-46022 Valencia, Spain\\
 \email{\{dinsa,jsilva\}@dsic.upv.es}
}

\maketitle


\begin{abstract}
\noindent

In software engineering, taking a good election between recursion and iteration is essential because their efficiency and maintenance are different. In fact, developers often need to transform iteration into recursion (e.g., in debugging, to decompose the call graph into iterations);  thus, it is quite surprising that there does not exist a public transformation from loops to recursion that handles all kinds of loops.
This article describes a transformation able to transform iterative loops into equivalent recursive methods. The transformation is described for the programming language Java, but it is general enough as to be adapted to many other languages that allow iteration and recursion.
We describe the changes needed to transform loops of types \emph{while/do/for/foreach} into recursion. 
Each kind of loop requires a particular treatment that is described and exemplified.
\end{abstract}

\keywords{Program transformation, Iteration, Recursion}


\section{Introduction}
\label{sec_introduction}


Iteration and recursion are two different ways to reach the same objective. 
In some paradigms, such as the functional or logic, iteration does not even exist. 
In other paradigms, e.g., the imperative or the object-oriented paradigm, the programmer can decide which of them to use. 
However, they are not totally equivalent, and sometimes it is desirable to use recursion, while other times iteration is preferable. 
In particular, one of the most important differences is the performance achieved by both of them. 
In general, compilers have produced more efficient code for iteration, and this is the reason 
why several transformations from recursion to iteration exist (see, e.g., \cite{HarK92,LiuS00,Mcc62}). 
Recursion in contrast is known to be more intuitive, reusable and debuggable. 
In fact, other researchers have obtained both theoretical and experimental results showing significant performance benefits of recursive algorithms on both uniprocessor hierarchies and on shared-memory systems \cite{YiAK00}. In particular, Gustavson and Elmroth \cite{Gus97,ElmG00} have demonstrated significant performance benefits from recursive versions of Cholesky and QR factorization, and Gaussian elimination with pivoting. 

Transforming loops to recursion is also useful in debugging, as demonstrated by the technique presented in \cite{InsST12c}. In this paper, transforming all iterative loops into recursive methods before starting an algorithmic debugging session can improve the interaction between the debugger and the programmer, and it can also reduce the granularity of the errors found. In particular, algorithmic debuggers only report buggy methods. Thus, a bug inside a loop is reported as a bug in the whole method that contains the loop, which is sometimes too imprecise. Transforming a loop into a recursive method allows the debugger to identify the recursive method (and thus the loop) as buggy. 
Hence, we wanted to implement this transformation and integrate it in the \emph{Declarative Debugger for Java} (DDJ) \cite{InsS10}, but, surprisingly, we did not find any 
available transformation from iterative loops into recursive methods for Java (neither for any other object-oriented language). Therefore, we had to implement it by ourselves and decided 
to automatize and generalize the transformation to make it publicly available. From the best of our knowledge this is the first transformation for all kinds of iterative loops. 

One important property of our transformation is that it always produces tail recursive methods \cite{Cli98}. 
This means that they can be compiled to efficient code because the compiler only needs to keep two activation records in the stack to execute the whole loop \cite{Han90,Bak96}.
Another important property is that each iteration is always represented with one recursive call. This means that a loop that performs 100 iterations is transformed into a recursive method that performs 100 recursive calls. This equivalence between iterations and recursive calls is very important for some applications such as debugging, and it produces code that is more maintainable.  

In the rest of the paper we describe our transformation for all kinds of loops in Java (i.e., \emph{while/do/for/foreach}). 
The transformation of each particular kind of loop is explained with an example. 
We start with an illustrative example that provides the reader with a general view of how the transformation works.

\begin{example}
Consider the Java code in Algorithm~\ref{alg_sqrtWhile} that computes the square root of the input argument.

\begin{algorithm}
\caption{Sqrt (iterative version)}
\label{alg_sqrtWhile}
\begin{algorithmic}[1]
    \STATE {\bf public double} sqrt({\bf double} x) \{
    \STATE ~~~~{\bf if} (x $<$ 0)
    \STATE ~~~~~~~~{\bf return} Double.NaN;
    \STATE ~~~~{\bf double} b = x;
    \STATE ~~~~{\bf while} (Math.abs(b * b - x) $>$ 1e-12)
    \STATE ~~~~~~~~b = ((x / b) + b) / 2;
    \STATE ~~~~{\bf return} b;
    \STATE \}
\end{algorithmic}
\end{algorithm}

\noindent This algorithm implements a \emph{while}-loop where each iteration obtains a more accurate approximation of the square root of variable \emph{x}. 
The transformed code is depicted in Algorithm~\ref{alg_sqrtWhileTransformed} that implements the same functionality but replacing the \emph{while}-loop with a new recursive method $sqrt\_loop$.

\begin{algorithm}
\caption{Sqrt (recursive version)}
\label{alg_sqrtWhileTransformed}
\begin{algorithmic}[1]
    \STATE {\bf public double} sqrt({\bf double} x) \{
    \STATE ~~~~{\bf if} (x $<$ 0)
    \STATE ~~~~~~~~{\bf return} Double.NaN;
    \STATE ~~~~{\bf double} b = x;
    \STATE ~~~~{\bf if} (Math.abs(b * b - x) $>$ 1e-12)
    \STATE ~~~~~~~~b = this.sqrt\_loop(x, b);
    \STATE ~~~~{\bf return} b;
    \STATE \}
    \STATE {\bf private double} sqrt\_loop({\bf double} x, {\bf double} b) \{
    \STATE ~~~~b = ((x / b) + b) / 2;
    \STATE ~~~~{\bf if} (Math.abs(b * b - x) $>$ 1e-12)
    \STATE ~~~~~~~~{\bf return} this.sqrt\_loop(x, b);
    \STATE ~~~~{\bf return} b;
    \STATE \}
\end{algorithmic}
\end{algorithm}
\end{example}

Essentially, the transformation performs two steps:
\begin{enumerate}
\item Substitute the original loop by new code (lines 5-6 in Algorithm~\ref{alg_sqrtWhileTransformed}).
\item Create a new recursive method (lines 9-14 in Algorithm~\ref{alg_sqrtWhileTransformed}).
\end{enumerate}

In Algorithm~\ref{alg_sqrtWhileTransformed}, the new code in method $sqrt$ includes a call (line 6) to the recursive method $sqrt\_loop$ that implements the loop (lines 9-14). 
This new recursive method contains the body of the original loop (line 10). 
Therefore, each time the method is invoked, an iteration of the loop is performed. 
The rest of the code added during the transformation (lines 5, 11-13) is the code needed to simulate the same effects of a \emph{while}-loop. 
Therefore, this is the only code that we should change to adapt the transformation to the other kinds of loops (\emph{do/for/foreach}).

\section{Transforming loops into recursive methods}
\label{sec-transf}

\begin{table*}[p]
\centering
\includegraphics[width=16cm]{\figures{tablaResumen.png}}
\caption{Loops transformation taxonomy}
\label{fig_tablaResumen}
\end{table*}

Our program transformations are summarized in Table~\ref{fig_tablaResumen}. 
This table has a different row for each kind of loop. 
For each loop, we have two columns. 
One for the iterative version of the loop, and one for the transformed recursive version.  
Observe that the code is presented in an abstract way, so that it is formed by a parameterized skeleton of the code that can be instantiated with any particular loop of each kind.

In the recursive version, the code inside the ellipses is code inserted by the programmer (it comes from the iterative version). 
The rest of the code is automatically generated by the transformation. 
Here, {\tt result} and {\tt loop} are fresh names (not present in the iterative version) for a variable and a method respectively; {\tt type} is a data type that corresponds to the data type declared by the user (it is associated to a variable already declared in the iterative version).
The code inside the squares has the following meaning:
\begin{description}
\item \fbox{1} contains the sequence formed by all variables declared in {\tt Code1} (and in {\tt ini} in \emph{for}-loops) that are used in {\tt Code2} and {\tt cond} (and in {\tt upd} in \emph{for}-loops).  
\item \fbox{1'} contains the previous sequence but including types (because it is used as the parameters of the method, and the previous sequence is used as the arguments of the call to the method). 
\item \fbox{2} contains for each object in the array {\tt result} (which contains the same variables as \fbox{1} and \fbox{1'}), a casting of the object to assign the corresponding type. For instance, if the array contains two variables [{\tt x,y}] whose types are respectively {\tt double} and {\tt int};  then \fbox{2} contains:\\ {\tt x = (Double) result[0];\\ y = (Integer) result[1]};
\end{description}
 
Observe that, even though these steps are based on Java, the same steps (with small modifications) can be used to transform loops in many other imperative or object-oriented languages. 
The code in Table~\ref{fig_tablaResumen} is generic. In some specific cases, this code can be optimized. 
For instance, observe that the recursive method always returns an array of objects ({\tt return new Object[] \{...\}}) with all variables that changed in the loop. This array is unnecessary and inefficient if the recursive method only needs to return one variable (or if it does not need to return any variable). Therefore, the creation of the array should be replaced by a single variable or null (i.e., {\tt return null}).  
In the rest of the paper, we always apply optimizations when possible, so that the code does not perform any unnecessary operations. This allows us to present a generic transformation as the one in Table~\ref{fig_tablaResumen}, and also to provide specific efficient transformations for each kind of loop. 
The optimizations are not needed to understand the transformation, but they should be considered when implementing it. 
In the rest of this section we explain the transformation of all kinds of loop.
The four kinds of loops (\emph{while/do/for/foreach}) present in the Java language behave nearly in the same way. Therefore, the modifications needed to transform each kind of loop into a recursive method are very similar. We start by describing the transformation for \emph{while}-loops, and then we describe the variations needed to adapt the transformation for \emph{do/for/foreach}-loops.

\subsection{Transformation of \emph{while}-loops}
\label{sec_loopTransformation}

\begin{figure}[h!]
\centering
\subfigure[Original method]
{\includegraphics[width=6cm]{\figures{genericLoop.png}}}
\subfigure[Transformed method]
{\includegraphics[width=10cm]{\figures{genericRecursiveMethodCaller.png}}\label{fig_genericRecursiveMethodCaller}}
\\
\subfigure[Recursive method]
{\includegraphics[width=9cm]{\figures{genericRecursiveMethod.png}}\label{fig_genericRecursiveMethod}}
\caption{\emph{while}-loop transformation}
\label{fig_whileSubstitution}
\end{figure}

In Table~\ref{tab_pasos} we show a general overview of the steps needed to transform a Java iterative \emph{while}-loop into an equivalent recursive method. Each step is described in the following.

\begin{table*}[h!]
\small
\begin{center}
\begin{tabular}{l@{~~}|@{~~}c}
{\bf Step} & {\bf Correspondence with Figure~\ref{fig_whileSubstitution}}\\
\hline
~& {\bf Figure~\ref{fig_genericRecursiveMethodCaller}}\\
\hspace{-1em}1) ~~~Substitute the loop by a call to the recursive method & Caller\\
\hspace{-1em}1.1) ~~~~~If the loop condition is satisfied & Loop condition\\
\hspace{-1em}1.1.1) ~~~~~~~Perform the first iteration & First iteration\\
\hspace{-1em}1.2) ~~~~~Catch the variables modified during the recursion & Modified variables\\
\hspace{-1em}1.3) ~~~~~Update the modified variables & Updated variables\\
~&\\
~& {\bf Figure~\ref{fig_genericRecursiveMethod}}\\
\hspace{-1em}2) ~~~Create the recursive method & Recursive method\\
\hspace{-1em}2.1) ~~~~~Define the method's parameters & Parameters\\
\hspace{-1em}2.2) ~~~~~Define the code of the recursive method & \\
\hspace{-1em}2.2.1) ~~~~~~~Include the code of the original loop & Loop code\\
\hspace{-1em}2.2.2) ~~~~~~~If the loop condition is satisfied & Loop condition\\
\hspace{-1em}2.2.2.1) ~~~~~~~~~Perform the next iteration & Next iteration\\
\hspace{-1em}2.2.3) ~~~~~~~Otherwise return the modified variables & Modified variables\\
\end{tabular}
\end{center}
\caption{Steps of the \emph{while}-loop transformation}
\label{tab_pasos}
\end{table*}


\subsubsection{Substitute the loop by a call to the recursive method}
The first step is to remove the original loop and substitute it with a call to the new recursive method. 
We can see this substitution in Figure~\ref{fig_genericRecursiveMethodCaller}. 
Observe that some parts of the transformation have been labeled to ease later references to the code.
The tasks performed during the substitution are explained in the following:



\begin{itemize}
\item {\bf Perform the first iteration}\\
In the \emph{while}-loop, first of all we check whether the \emph{loop condition} holds. 
If it does not hold, then the loop is not executed. 
Otherwise, the \emph{first iteration} is performed by calling the recursive method with the variables used inside the loop as arguments of the method call. Hence, we need an analysis to know what variables are used inside the loop.
The recursive method is in charge of executing as many iterations of the loop as needed.

\item {\bf Catch the variables modified during the recursion}\\
The variables modified during the recursion cannot be automatically updated in Java because all parameters are passed by value. 
Therefore, if we modify an argument inside a method we are only modifying a copy of the original variable. 
This also happens with objects. 
Hence, in order to output those  \emph{modified variables} that are needed outside the loop, we use an array of objects.
Because the \emph{modified variables} can be of any data type\footnote{In the case that the returned values are primitive types, then they are naturally encapsulated by the compiler in their associated primitive wrapper classes. }, we use an array of objects of class {\bf Object}.

In presence of call-by-reference, this step should be omitted.

\item {\bf Update the modified variables}\\
After the execution of the loop, the \emph{modified variables} are returned inside an {\bf Object} array.
Each variable in this array must be cast to its respective type before being assigned to the corresponding variable declared before the loop.

In presence of call-by-reference, this step should be omitted.
\end{itemize}

\subsubsection{Create the recursive method}
Once we have substituted the loop, we create a new method that implements the loop in a recursive way. 
This \emph{recursive method} is shown in Figure~\ref{fig_genericRecursiveMethod}.

The code of the \emph{recursive method} is explained in the following:

\begin{itemize}
\item {\bf Define the method's parameters}\\
There are variables declared inside a method but declared outside the loop and used by this loop.  
When the loop is transformed into a \emph{recursive method}, these variables are not accessible from inside the \emph{recursive method}. 
Therefore, they must be passed as arguments in the calls to it.
Hence, the parameters of the \emph{recursive method} are the intersection between the variables declared before the loop and the variables used inside it.

\item {\bf Define the code of the recursive method}\\
Each iteration of the original iterative loop is emulated with a call to the new recursive method. 
Therefore in the code of the \emph{recursive method} we have to execute the current iteration and control whether the \emph{next iteration} must be executed or not.

\begin{itemize}
\item {\bf Include the code of the original loop}\\
When the \emph{recursive method} is invoked it means that we want to execute one iteration of the loop. Therefore, we place the \emph{original code} of the loop at the beginning of the \emph{recursive method}. This code is supposed to update the variables that control the \emph{loop condition}. Otherwise, the original loop is in fact an infinite loop and the recursive method created will be invoked infinitely.

\item {\bf Perform the next iteration}\\
Once the iteration is executed, we check the \emph{loop condition} again to know whether another iteration must still be executed. 
In such a case, we perform the \emph{next iteration} with the same arguments. 
Note that the values of the arguments can be modified during the execution of the iteration, therefore, each iteration has different arguments values, but the names and the number of arguments remain always the same.

\item {\bf Otherwise return the modified variables}\\
If the loop condition does not hold, the loop ends and thus we must finish the sequence of \emph{recursive method} calls and return to the original method in order to continue executing the rest of the code. 
Because the arguments have been updated in each recursive call, at this point we have the last values of the variables involved in the loop. 
Hence these variables must be returned in order to update them in the original method. 
Observe that these variables are passed from iteration to iteration during the execution of the recursive method until it is finally returned to the recursive method caller. 

In presence of call-by-reference, this step should be omitted.
\end{itemize}
\end{itemize}

Figure~\ref{fig_whileSubstitutionInstance} shows an example of transformation of a \emph{while}-loop.

\begin{figure}[h!]
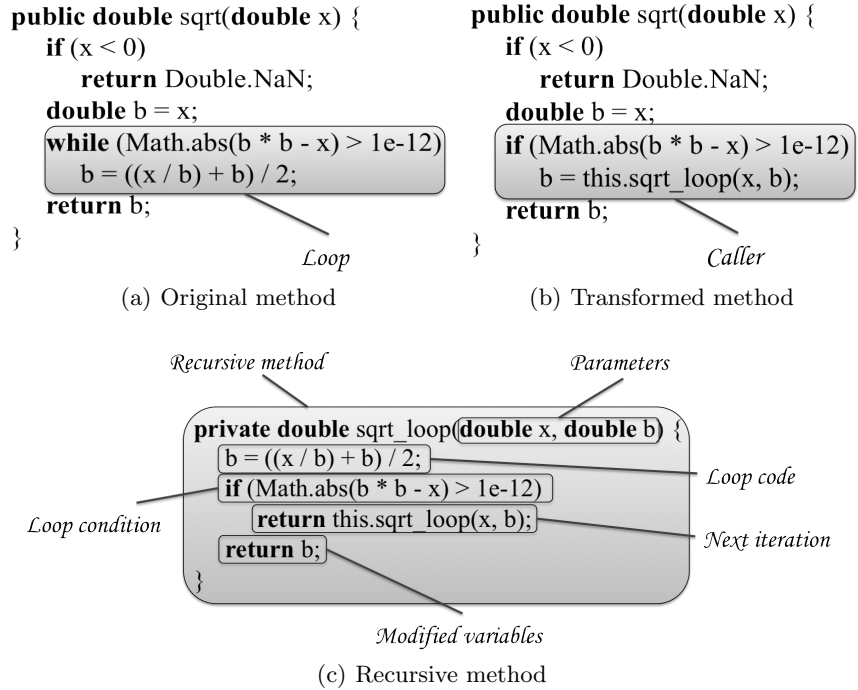

\centering
\subfigure[Original method]
{\includegraphics[width=6cm]{\figures{whileLoop.png}}}
\subfigure[Transformed method]
{\includegraphics[width=5.3cm]{\figures{whileRecursiveMethodCaller.png}}}
\subfigure[Recursive method]
{\includegraphics[width=11cm]{\figures{whileRecursiveMethod.png}}}
\caption{\emph{while}-loop transformation}
\label{fig_whileSubstitutionInstance}
\end{figure}

\subsection{Transformation of \emph{do}-loops}
\label{sec_doLoop}

%

\emph{do}-loops behave exactly in the same way as \emph{while}-loops except in one detail: The first iteration of the \emph{do}-loop is always performed. In Figure~\ref{fig_doLoop} we can see an example of a \emph{do}-loop.

\begin{figure}[h!]
\centering
\includegraphics[width=5cm]{\figures{doLoop.png}}
\caption{do-loop}
\label{fig_doLoop}
\end{figure}

This code obtains the square root value of variable \emph{x} as the code in Algorithm~\ref{alg_sqrtWhile}. 
The difference is that, if variable \emph{x} is either 0 or 1, then the method directly returns variable \emph{x}, otherwise the loop is performed in order to calculate the square root. 
In order to transform the \emph{do}-loop into a recursive method, we can follow the same steps used in Table~\ref{tab_pasos} with only one change: in step 1.1 the loop condition is not evaluated; instead, we only need to add a new code block to ensure that those variables created during the transformation are not available outside the transformed code.

Figure~\ref{fig_doTransformation} illustrates the only change needed to transform the \emph{do}-loop into a recursive method. Observe that in this example there is no need to introduce a new block, because the transformed code does not create new variables, but in the general case the block could be needed.

\begin{figure}[h!]
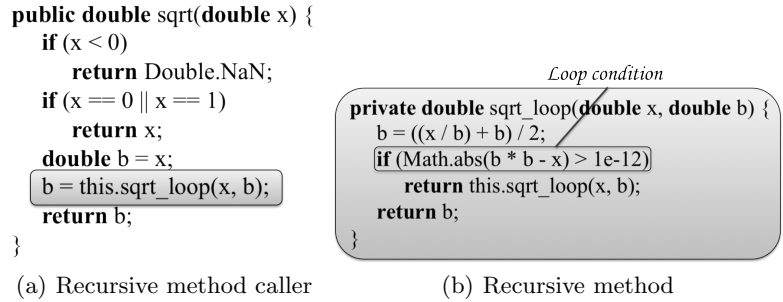

\centering
\subfigure[Recursive method caller]
{\includegraphics[width=4.25cm]{\figures{doRecursiveMethodCaller.png}}}
\subfigure[Recursive method]
{\includegraphics[width=6cm]{\figures{doRecursiveMethod.png}}}
\caption{do-loop transformation}
\label{fig_doTransformation}
\end{figure}
 
\begin{itemize}
\item {\bf Add a new code block}\\
Observe in Table~\ref{fig_tablaResumen}, in column {\tt Caller}, that, contrarily to \emph{while}-loops, \emph{do}-loops need to introduce a new \emph{block} (i.e., a new scope).
The reason is that there could exist variables with the same name as the variables created during the transformation (e.g., \emph{result}). 
Hence, the new block avoids variable clashes and limits the scope of the variables created by the transformation. 
\end{itemize}

\subsection{Transformation of \emph{for}-loops}
\label{sec_forLoop}

One of the most frequently used loops in Java is the \emph{for}-loop. 
This loop behaves exactly in the same way as the \emph{while}-loop except in one detail: 
\emph{for}-loops provide the programmer with a mechanism to declare, initialize and update variables that will be accessible inside the loop. 

In Figure~\ref{fig_forLoop} we can see an example of a \emph{for}-loop. 
This code obtains the square root value of variable \emph{x} exactly as the code in Algorithm~\ref{alg_sqrtWhile}, 
but it also prints the approximation obtained in every iteration. 
We can see in Figure~\ref{fig_forRecursiveMethodCaller} and \ref{fig_forRecursiveMethod} the additional changes needed to transform the \emph{for}-loop into a recursive method.

\begin{figure}[h!]
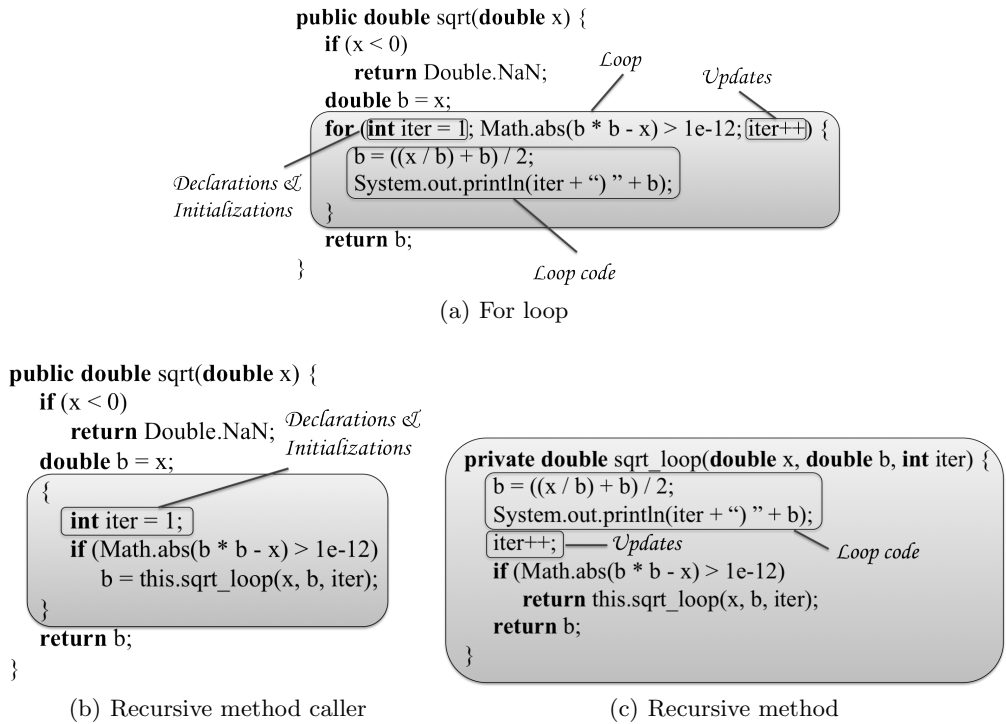

\centering
\subfigure[For loop]
{\includegraphics[width=9cm]{\figures{forLoop.png}}\label{fig_forLoop}}
\\
\subfigure[Recursive method caller]
{\includegraphics[width=5.75cm]{\figures{forRecursiveMethodCaller.png}}\label{fig_forRecursiveMethodCaller}}
\subfigure[Recursive method]
{\includegraphics[width=7.5cm]{\figures{forRecursiveMethod.png}}\label{fig_forRecursiveMethod}}
\caption{for-loop transformation}
\label{fig_forTransformation}
\end{figure}

As shown in Figure~\ref{fig_forTransformation}, in order to transform the \emph{for}-loop into a recursive method, we can follow the same steps used in Table~\ref{tab_pasos}, but we have to make three changes:\\
\begin{itemize}
\item {\bf Add a new code block}\\
Exactly in the same way and with the same purpose as in \emph{do}-loops. 
\item {\bf Add the declarations/initializations at the beginning of the block}\\
In the original method, those variables created during the \emph{declaration} and \emph{initialization} of the loop are only available inside it (and not in the code that follows the loop). We must ensure that these variables keep the same scope in the transformed code. 
This can be easily achieved with the new \emph{block}. 
In the transformed code, those variables are declared and initialized at the beginning of the new \emph{block}, and they are passed as arguments to the recursive method in every iteration to make them accessible inside it. 

\item {\bf Add the updates between the loop code and the loop condition}\\
In \emph{for}-loops there exists the possibility of executing code between iterations. 
This code is usually a collection of \emph{updates} of the variables declared at the beginning of the loop (e.g., in Figure~\ref{fig_forLoop} this code is {\tt iter++}). 
However, this code could be formed by a series of expressions separated by commas that could include method invocations, assignments, etc. 
Because this \emph{update} code is always executed before the condition of the loop, it must be placed in the recursive method between the \emph{loop code} and the \emph{loop condition}.
\end{itemize}

\subsection{Transformation of \emph{foreach}-loops}
\label{sec_foreachLoop}

\emph{foreach}-loops are specially useful to traverse collections of elements. 
In particular, this kind of loops traverses a given collection and it executes a block of code for each element. 
The transformation of a \emph{foreach}-loop into a recursive method is different depending on the kind of collection that is traversed. 
In Java we can use \emph{foreach}-loops either with \emph{arrays} or \emph{iterable} objects. 
We explain each transformation separately.


\subsubsection{\emph{foreach}-loop used to traverse arrays}
An array is a composite data structure where elements have been sequentialized, and thus, they can be traversed linearly.
We can see an example of a \emph{foreach}-loop that traverses an array in Algorithm~\ref{alg_foreachArray}.

\begin{algorithm}[h!]
\caption{\emph{foreach}-loop that traverses an array (iterative version)}
\label{alg_foreachArray}
\begin{algorithmic}[1]
    \STATE {\bf public void} foreachArray() \{
    \STATE ~~~~{\bf double[]} numbers = {\bf new double[]} \{ 4.0, 9.0 \};
    \STATE ~~~~{\bf for} ({\bf double} number : numbers) \{
    \STATE ~~~~~~~~{\bf double} sqrt = this.sqrt(number);
    \STATE ~~~~~~~~System.out.println(``sqrt('' + number + ``) = '' + sqrt);
    \STATE ~~~~\}
    \STATE \}
\end{algorithmic}
\end{algorithm}

This code computes and prints the square root of all elements in the array [4.0, 9.0]. 
Each individual square root is computed with Algorithm~\ref{alg_sqrtWhile}.
The \emph{foreach}-loop traverses the array sequentially starting in position \emph{numbers[0]} until the last element in the array. 
The transformation of this loop into an equivalent recursive method is very similar to the transformation of a \emph{for}-loop. 
However there are differences. For instance, \emph{foreach}-loops lack of a counter. 
This can be observed in Figure~\ref{fig_foreachArrayTransformation} that implements a recursive method equivalent to the loop in Algorithm~\ref{alg_foreachArray}.


\begin{figure*}[t!]
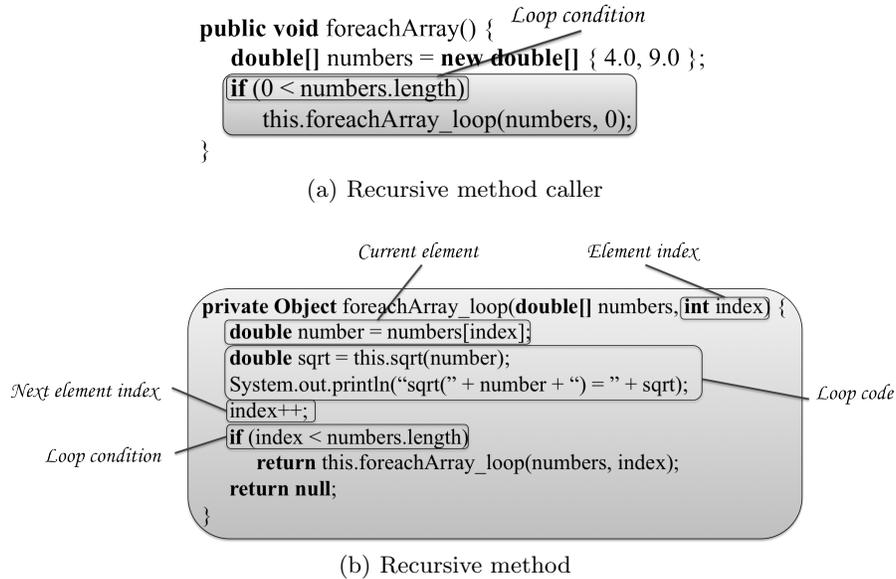

\centering
\subfigure[Recursive method caller]
{\includegraphics[width=7cm]{\figures{foreachArrayRecursiveMethodCaller.png}}}
\\
\subfigure[Recursive method]
{\includegraphics[width=12cm]{\figures{foreachArrayRecursiveMethod.png}}}
\caption{\emph{foreach}-loop transformation (Array version)}
\label{fig_foreachArrayTransformation}
\end{figure*}

In Figure~\ref{fig_foreachArrayTransformation} we can see the symmetry with respect to the \emph{for}-loop transformation. 
The only difference is the creation of a fresh variable that is passed as argument in the recursive method calls (in the example this variable is called \emph{index}). This variable is used for:

\begin{itemize}
\item {\bf Controlling whether there are more elements to be treated}\\
A \emph{foreach}-loop is only executed if the array contains elements. 
Therefore we need a \emph{loop condition} in the recursive method caller and another in the recursive method to know when there are no more elements in the array and thus finish the traversal.
The later is controlled with a variable (\emph{index} in the example) acting as a counter.
\item {\bf Obtaining the next element to be treated}\\
During each iteration of the \emph{foreach}-loop a variable called \emph{number} is instantiated with one element of the array (line 3 of Algorithm~\ref{alg_foreachArray}). 
In the transformation this behavior is emulated by declaring and initializing this variable at the beginning of the recursive method. 
It is initialized to the corresponding element of the array by using variable \emph{index}.
\end{itemize}

\subsubsection{\emph{foreach}-loop used to traverse \emph{iterable} objects}

A \emph{foreach}-loop can be used to traverse objects that implement the interface \emph{Iterable}.
Algorithm~\ref{alg_foreachIterable} shows an example of a \emph{foreach}-loop using one of these objects.

\begin{algorithm}[h!]
\caption{\emph{foreach}-loop used to traverse an iterable object (iterative version)}
\label{alg_foreachIterable}
\begin{algorithmic}[1]
    \STATE {\bf public void} foreachIterable() \{
    \STATE ~~~~{\bf List$<$Double$>$} numbers = Arrays.asList(4.0, 9.0);
    \STATE ~~~~{\bf for} ({\bf double} number : numbers) \{
    \STATE ~~~~~~~~{\bf double} sqrt = this.sqrt(number);
    \STATE ~~~~~~~~System.out.println(``sqrt('' + number + ``) = '' + sqrt);
    \STATE ~~~~\}
    \STATE \}
\end{algorithmic}
\end{algorithm}

This code behaves exactly in the same way as Algorithm~\ref{alg_foreachArray} but using an \emph{iterable} object instead of an array (\emph{numbers} is an iterable object because it is an instance of class \emph{List} that in turn implements the interface \emph{Iterable}). 
The interface \emph{Iterable} only has one method, called \emph{iterator}, that returns an object that implements the \emph{Iterator} interface.
With regard to the interface \emph{Iterator}, it forces the programmer to implement the \emph{next}, \emph{hasNext} and \emph{remove} methods; 
and these methods allow the programmer to freely implement how the collection is traversed (e.g., the order, whether repetitions are taken into account or not, etc.). Therefore, the transformed code should use these methods to traverse the collection.
We can see in Figure~\ref{fig_foreachIterableTransformation} a recursive method equivalent to Algorithm~\ref{alg_foreachIterable}.


\begin{figure*}[t!]
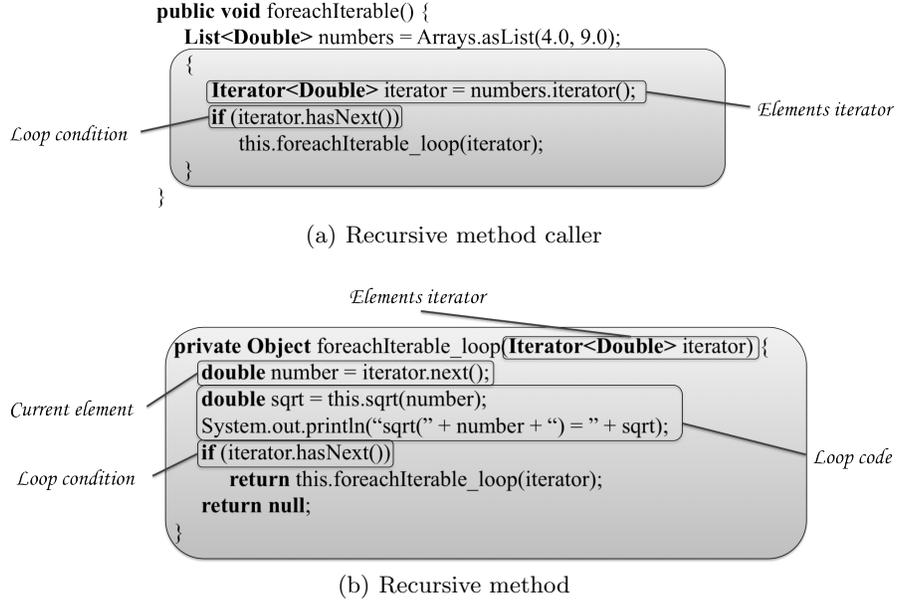

\centering
\subfigure[Recursive method caller]
{\includegraphics[width=12cm]{\figures{foreachIterableRecursiveMethodCaller.png}}}
\\
\subfigure[Recursive method]
{\includegraphics[width=12cm]{\figures{foreachIterableRecursiveMethod.png}}}
\caption{foreach-loop transformation (Iterable version)}
\label{fig_foreachIterableTransformation}
\end{figure*}

Observe that the transformed code in Figure~\ref{fig_foreachIterableTransformation} is very similar to the one in Figure~\ref{fig_foreachArrayTransformation}. The only difference is the use of an \emph{iterator} variable (instead of an integer variable) that controls the element of the collection to be treated. Note that method \emph{next} of variable \emph{iterator} allows us to know what is the next element to be treated, and method \emph{hasNext} tell us whether there exist more elements to be processed yet.


\section{Correctness\label{sec:proofs}}

In this section we provide a formal semantics-based specification of our transformation in order to prove 
its correctness.
For this, we provide a BNF syntax specification and an operational semantics of Java. We consider the subset of Java that is needed to implement the transformation (\emph{if-then-else}, \emph{while}, method calls, \emph{return}, etc.), and we ignore the rest of syntax elements for the sake of simplicity (they do not have any influence because any command inside the body of the loop remains unchanged in the transformed code).  
Moreover, in this section, we center the discussion on \emph{while}-loops and prove properties for this kind of loop. The proof for the other kinds of loops is omitted, but it would be in all cases analogous or slightly incremental. 

We start with a BNF syntax of Java:

%

 
\noindent\rule{\linewidth}{.01in} {\small
\[
\begin{array}{lcl@{~~~~}l@{~~~}l}
P & ::= & M_1, \ldots, M_n, S_p & \mbox{ (program) } & Domains \\ 
  &     &     &     & x,y,z\ldots \in \mathbb{V} ~\mbox{\sf (variables)} \\
  &     &     &     & a,b,c\ldots \in \mathbb{C} ~\mbox{\sf (constants)} \\

M & ::= & m(x_1, \dots, x_n)~\{~S_p;~S_r~\} & \mbox{ (method definition) } & $where $x_1,\ldots,x_n \in \mathbb{V}$ $\\ 
  &     &     &     & $and $m$ is the name$\\
  &     &     &     & $of the method$\\

S_p & ::= & x := E  & \mbox{ (assignment) } \\ 
  & | & x := m(E_0, \dots, E_n)  & \mbox{ (method invocation) } \\
  & | & $if $E_b$ then $S_p  & \mbox{ (if-then) } \\
  & | & $if $E_b$ then $S_p$ else $S_p  & \mbox{ (if-then-else) } \\
  & | & $while $E_b$ do $S_p  & \mbox{ (while) } \\
  & | & S_p;~S'_p  & \mbox{ (sequence) } \\
 
S_r & ::= & $return $E  & \mbox{ (return) } \\ 

E & ::= & E_a \mid E_b  & \mbox{ (expresion) } \\ 
E_a & ::= & E_a + E_a \mid E_a - E_a \mid V & \mbox{ (arithmetic expresion) } \\   
E_b & ::= & E_b $ != $ E_b \mid E_b == E_b \mid V & \mbox{ (boolean expresion) } \\ 
V & ::= & x \mid a  & \mbox{ (variables or constants) } \\ 
\end{array}
\]
} \rule{\linewidth}{.01in}

\bigskip 
\noindent A program is a set of method definitions and at least one initial statement (usually a method invocation). 
Each method definition is composed of a set of statements followed by a \emph{return} statement.  
For simplicity, the arguments of a method invocation can only be expressions (not statements). 
This is not a restriction, because any statement can be assigned to a variable and then be passed as argument of the method invocation. However, this simplification allows us to ease the semantics of method invocations and, thus, it increases readability.  

In the following we consider two versions of the same program shown in Algorithms~\ref{alg_code1} and~\ref{alg_code2}. We assume that in Algorithm~\ref{alg_code1} there exists a variable $x$ already defined before the loop and, for the sake of simplicity, it is the only variable modified inside $S$. Therefore, Algorithm~\ref{alg_code2} is the recursive version of the \emph{while}-loop in Algorithm~\ref{alg_code1} according to our transformation, and hence, $p_0, \dots, p_n$ represent all variables defined before the loop and used in $S$ (the loop statements) and $cond$ (the loop condition). In the case that more than one variable are modified, then the output would be an array with all variables modified. We avoid this case because it is not necessary for the proof.
 
\begin{algorithm}[h!]
\caption{While version}
\label{alg_code1}
\footnotesize
\begin{algorithmic}[1]
    \STATE $~{\bf while}$ $cond$ ${\bf do}$
    \STATE $~~~~~S;$\\
\end{algorithmic}
\end{algorithm} 


\begin{algorithm}[h!]
\caption{Recursive version}
\label{alg_code2}
\footnotesize
\begin{algorithmic}[1]
    \STATE $~m(p_0, \dots, p_n)$ $\{$
    \STATE $~~~~~S;$
    \STATE $~~~~~{\bf if}$ $cond$ ${\bf then}$
    \STATE $~~~~~~~~~x := m(a_0, \dots, a_n);$
    \STATE $~~~~~{\bf return}$ $x;$
    \STATE $~\}$
    \STATE $~{\bf if}$ $cond$ ${\bf then}$
    \STATE $~~~~~x := m(a_0, \dots, a_n);$\\
\end{algorithmic}
\end{algorithm}

In order to provide an operational semantics for this Java subset, which allows recursion, we need a stack to push and pop different frames that represent individual method activations.  
Frames, $f_0,f_1,\ldots\in\mathbb{F}$, are sequences of pairs variable-value.
States, $s_0,s_1,\ldots\in\mathbb{S}$, are sequences of frames ($\mathbb{S}:~\mathbb{F}$ x $\dots$ x $\mathbb{F}$).
We make the program explicitly accessible to the semantics through the use of an environment, $e \in \mathbb{E}$, represented with a sequence of functions from method names $\mathbb{M}$ to pairs of parameters $\mathbb{P}$ and statements $\mathbb{I}$ ($\mathbb{E}:~(\mathbb{M} \rightarrow (\mathbb{P}$ x $\mathbb{I}))$ x $\dots$ x $(\mathbb{M} \rightarrow (\mathbb{P}$ x $\mathbb{I}))$). Our semantics is based on the Java semantics described in \cite{Nie92} with some minor modifications.
It uses a set of functions to update the state, the environment, etc. 

Function $\mathit{Upd_v}$ is used to update a variable ($var$) in the current frame of the state ($s$) with a value ($value$). The current frame in the state is always the last frame introduced (i.e., the last element in the sequence of frames that represent the state). We use the standard notation $f[var \rightarrow value]$ to denote that variable $var$ in frame $f$ is updated to value $value$.

\medskip
$\mathit{Upd_v}(s,~var \rightarrow value) =
\left \{
\begin{array}{lll}
& error & \mathit{if}~s = [ ]\\
& [ f_0, \dots, f_n[var \rightarrow value] ] & \mathit{if}~s = [ f_0, \dots, f_n ]\\
\end{array}
\right.$\\

Function $\mathit{Upd_r}$ records the returned value ($value$) of the current frame of the state ($s$) inside a fresh variable $\Re$ of this frame, so that other frames can consult the value returned by the current frame. 

\medskip
$\mathit{Upd_r}(s,~value) =
\left \{
\begin{array}{lll}
& error & \mathit{if}~s = [ ]\\
& [ f_0, \dots, f_n[\Re \rightarrow value] ] & \mathit{if}~s = [ f_0, \dots, f_n ]\\
\end{array}
\right.$\\

Function $\mathit{Upd_{vr}}$ is used to update a variable ($var$) in the penultimate frame of the state ($s$) taking the value returned by the last frame in the state (which must be previously stored in $\Re$). This happens when a method calls another method and the latter finishes returning a value. In this situation, the last frame in the state should be removed and the value returned should be updated in the penultimate frame. We use the notation $f_n(\Re)$ to consult the value of variable $\Re$ in frame $f_n$.

\medskip
$\mathit{Upd_{vr}}(s,~var) =
\left \{
\begin{array}{lll}
& error & \mathit{if}~s = [ ]~\mathit{or}~s = [ f ]\\
& [ f_0, \dots, f_{n-1}[var \rightarrow f_n(\Re)], f_n ] & \mathit{if}~s = [ f_0, \dots, f_{n-1}, f_n ]\\
\end{array}
\right.$\\

Function $\mathit{Upd_e}$ is used to update the environment ($env$) with a new method definition ($m \rightarrow (P, I)$). 
The environment is used in method invocations to know the method that should be executed.

\medskip
$\mathit{Upd_e}(env,~m \rightarrow (P, I)) = env[m \rightarrow (P, I)]$
\medskip

Function $\mathit{AddFrame}$ adds a new frame to the state ($s$). This frame is a sequence of mappings from parameters ($p_0, \dots, p_n$) to the evaluation of arguments ($a_0, \dots, a_n$). To evaluate an expression we use function $\mathit{Eval}$: a variable is consulted in the state, a constant is just returned, and a mathematical or boolean expression is evaluated with the standard semantics. We use this notation because the evaluation of expressions does not have influence in our proofs, but it significantly reduces the size of derivations, thus, improving clarity of presentation.

\medskip
$\mathit{AddFrame}(s, [p_0, \dots, p_n], [a_0, \dots, a_n]) =\\
~~~~~~~~~~~~~~~~~~~~~~~~~~~\left \{
\begin{array}{lll}
& [ [p_0 \rightarrow Eval(a_0), \dots, p_n \rightarrow Eval(a_n)] ] & if~s = [ ]\\
& [ f_0, \dots, f_m, [ p_0 \rightarrow Eval(a_0), \dots, p_n \rightarrow Eval(a_n) ] ] & if~s = [ f_0, \dots, f_m ]\\
\end{array}
\right.$\\


Analogously, function $\mathit{RemFrame}$ removes the last frame inserted into the state ($s$). 

\medskip
$RemFrame(s) =
\left \{
\begin{array}{lll}
& error & if~s = [ ]\\
& [ ] & if~s = [ f ]\\
& [ f_0, \dots, f_n] & if~s = [ f_0, \dots, f_n, f_{n + 1} ]\\
\end{array}
\right.$\\ 

%
%

We are now in a position ready to introduce our Java operational semantics. Essentially, the semantics is a big-step semantics composed of a set of rules of the form: $\frac{p_1\dots p_n}{\mathit{env}~\vdash~<st,~s>~\Downarrow~s'}$ that should be read as ``The execution of statement $st$ in state $s$ under the environment $env$ can be reduced to state $s'$ provided that  premises $p_1\dots p_n$ hold". The rules of the semantics are shown in Figure~\ref{fig:semantics}.

\begin{figure}[h!]
\noindent
\begin{center}
$\begin{array}{|c|}
\hline
\\
\mathit{New~method}\\
\frac{\mathit{env'}~=~\mathit{Upd_e}(env,~m_0~\rightarrow~(P,~I))~\land~\mathit{env'}~\vdash~<m_1~i,~s>~\Downarrow~s'~}{\mathit{env}~\vdash~<m_0(P) \{ I \}~m_1~i,~s>~\Downarrow~s'}\\
\\
\mathit{Empty~statement}\\
\frac{}{\mathit{env}~\vdash~<\surd,~s>~\Downarrow~s}\\
\\
\mathit{Asignment}\\
\frac{s'~=~\mathit{Upd_v}(s,~x~\rightarrow~\mathit{Eval}(op,~s))}{\mathit{env}~\vdash~<x:=op,~s>~\Downarrow~s'}\\
\\
\mathit{Method~invocation}\\
\frac{(P,~I)~=~\mathit{env}(m)~\land~s'~=~\mathit{AddFrame}(s, P, A)~\land~\mathit{env}~\vdash~<I,~s'>~\Downarrow~s''~\land~ s'''~=~\mathit{Upd_{vr}}(s'', x)~\land~s''''~=~\mathit{RemFrame}(s''')}{\mathit{env}~\vdash~<x:=m(A),~s>~\Downarrow~s''''}\\
\\
\mathit{If}\\
\begin{array}{c}
\frac{\mathit{env}~\vdash~<if~cond~then~i_0~else~\surd,~s>~\Downarrow~s'}{\mathit{env}~\vdash~<if~cond~then~i_0,~s>~\Downarrow~s'}\\
\\
\frac{<cond,~s>~\Rightarrow~true~\land~\mathit{env}~\vdash~<i_0,~s>~\Downarrow~s'}{\mathit{env}~\vdash~<if~cond~then~i_0~else~i_1,~s>~\Downarrow~s'}\\
\\
\frac{<cond,~s>~\Rightarrow~false~\land~\mathit{env}~\vdash~<i_1,~s>~\Downarrow~s'}{\mathit{env}~\vdash~<if~cond~then~i_0~else~i_1,~s>~\Downarrow~s'}\\
\end{array}\\
\\
\mathit{While}\\
\begin{array}{c}
\frac{<cond,~s>~\Rightarrow~false}{\mathit{env}~\vdash~<while~cond~do~i,~s>~\Downarrow~s}\\
\\
\frac{<cond,~s>~\Rightarrow~true~\land~\mathit{env}~\vdash~<i,~s>~\Downarrow~s'~\land~\mathit{env}~\vdash~<while~cond~do~i,~s'>~\Downarrow~s''}{\mathit{env}~\vdash~<while~cond~do~i,~s>~\Downarrow~s''}\\
\end{array}\\
\\
\mathit{Sequence}\\
\frac{\mathit{env}~\vdash~<i_0,~s>~\Downarrow~s'~\land~\mathit{env}~\vdash~<i_1,~s'>~\Downarrow~s''}{\mathit{env}~\vdash~<i_0;~i_1,~s>~\Downarrow~s''}\\
\\
\mathit{Return}\\
\frac{s'~=~\mathit{Upd_r}(s,~\mathit{Eval}(op,~s))}{\mathit{env}~\vdash~<\mathit{return}~op,~s>~\Downarrow~s'}\\
\\
\hline
\end{array}$
\end{center}
\caption{Java Semantics}
\label{fig:semantics}
\end{figure}

\bigskip

We can now prove our main result. 

\begin{theorem}[Correctness]
Algorithm~\ref{alg_code2} is semantically equivalent to Algorithm~\ref{alg_code1}.
\end{theorem}

\begin{proof}
We prove this claim by showing that the final state of Algorithm~\ref{alg_code1} is always the same as the final state of Algorithm~\ref{alg_code2}. The semantics of a program $P$ is:\\
\medskip

$S(P) = s~~$ iff $~~\mathit{[]}\vdash<P,~[]> \Downarrow s$\\

Therefore, we say that two programs $P_1$ and $P_2$ are equivalent if they have the same semantics:\\
\medskip

$S(P_1) = S(P_2)~~$ iff $~~\mathit{[]}\vdash<P_1,~[]> \Downarrow s_1 ~~\land~~ \mathit{[]}\vdash<P_2,~[]> \Downarrow s_2 ~~\land~~ s_1 = s_2$\\

\noindent For the sake of generality, in the following we consider that the loops can appear inside any other code. Therefore, the environment and the state are not necessarily empty. Thus, we will assume an initial environment $env_0$ and an initial state $s$: $\mathit{env_0}\vdash<P,~s>$.
We proof this semantic equivalence analyzing two possible cases depending on whether the loop is executed or not.\\

\noindent
{\bf 1) \underline{Zero iterations}}\\
\indent This situation can only happen when the condition $cond$ is not satisfied the first time it is evaluated. 
Hence, we have the following semantics derivation for each program:\\

\noindent \underline{Iterative version}\\
{\tiny \begin{center}
$\infer{\mathit{env}_0~\vdash~<while~cond~do~S,~s>~\Downarrow~s}{<cond,~s>~\Rightarrow~false}$\\
\end{center}}

\noindent \underline{Recursive version}\\
{\tiny \begin{center}
$\infer{\mathit{env}_0~\vdash~<m(P)~\{~S; I~\}~if~cond~then~t,~s>~\Downarrow~s} {
	\mathit{env} = \mathit{Upd_e}(\mathit{env}_0,~m~\rightarrow~(P,~S; I)) &
	\infer{\mathit{env}~\vdash~<if~cond~then~t,~s>~\Downarrow~s} {
		\infer{\mathit{env}~\vdash~<if~cond~then~t~else~\surd,~s>~\Downarrow~s} {
			<cond,~s>~\Rightarrow~false &
			\infer{\mathit{env}~\vdash~<\surd,~s>~\Downarrow~s}{}
		}
	}
}$
\end{center}}

\noindent
%
Clearly, the state is never modified neither in the iterative version nor in the recursive version. Therefore, both versions are semantically equivalent.\\


\noindent
{\bf 2) \underline{One or more iterations}}\\
\indent This means that the condition $cond$ is satisfied at least once. 
Let us consider that $cond$ is satisfied $n$ times, producing $n$ iterations. 
We proof that the final state of the program in Algorithm~\ref{alg_code2} is equal to the final state of the program in Algorithm~\ref{alg_code1} by induction over the number of iterations performed.\\

\noindent{\bf(Base Case)} In the base case, only one iteration is executed. 
Hence, we have the following derivations:\\

\noindent \underline{Iterative version}\\
{\tiny \begin{center}
$\infer{\mathit{env}_0~\vdash~<while~cond~do~S,~s>~\Downarrow~s^1}{
	<cond,~s>~\Rightarrow~true &
	\mathit{env}_0~\vdash~<S,~s>~\Downarrow~s^1 &
	\infer{\mathit{env}_0~\vdash~<while~cond~do~S,~s^1>~\Downarrow~s^1}{<cond,~s>~\Rightarrow~false}
}$
\end{center}}

\noindent \underline{Recursive version}\\
{\tiny \begin{center}
$\infer{\bigtriangleup} {
	\infer{\mathit{env}~\vdash~<S; I,~s^1>~\Downarrow~s^3} {
		\mathit{env}~\vdash~<S,~s^1>~\Downarrow~s^2 &
		\infer{\mathit{env}~\vdash~<if~cond~then~t;~return~x,~s^2>~\Downarrow~s^3} {
			\infer{\mathit{env}~\vdash~<if~cond~then~t,~s^2>~\Downarrow~s^2} {
				\infer{\mathit{env}~\vdash~<if~cond~then~t~else~\surd,~s^2>~\Downarrow~s^2} {
					<cond,~s^2>~\Rightarrow~false &
					\infer{\mathit{env}~\vdash~<\surd,~s^2>~\Downarrow~s^2} {}
				}
			} &
			\infer{\mathit{env}~\vdash~<return~x,~s^2>~\Downarrow~s^3} {s^3~=~\mathit{Upd_r}(s^2,~\mathit{Eval}(x, s^2))}
		}
	}
}$
\end{center}}

\noindent
{\tiny \begin{center}
$\infer{\mathit{env}~\vdash~<if~cond~then~t,~s>~\Downarrow~s^5} {
	\infer{\mathit{env}~\vdash~<if~cond~then~t~else~\surd,~s>~\Downarrow~s^5} {
		<cond,~s>~\Rightarrow~true &
		\infer{\mathit{env}~\vdash~<x:= m(a_0, \dots, a_n),~s>~\Downarrow~s^5} {
			(P, I) = env(m) &
			s^1~=~\mathit{AddFrame}(s, P, [a_0, \dots, a_n]) &
			\bigtriangleup &
			s^4~=~\mathit{Upd_{vr}}(s^3, x) &
			s^5~=~\mathit{RemFrame}(s^4)
		}
	}
}$
\end{center}}


\noindent
We can assume that variable $x$ has an initial value $z^0$, which must be the same in both versions of the algorithm. Then, states are modified during the iteration as follows:

{\scriptsize
\begin{center}
$
\begin{array}{lc@{~~~~~~~~~~~~~~~~~~~~~}l}
$\underline{Iterative version}$ & & $\underline{Recursive version}$\\
s = [ f_0 ] \Rightarrow f_0 = \{ x \rightarrow z^0 \} & & s = [ f_0 ] \Rightarrow f_0 = \{ x \rightarrow z^0 \}\\
s^1 = [ f_0 ] \Rightarrow f_0 = \{ x \rightarrow z^1 \} & & s^1 = [ f_0, f_1 ] \Rightarrow f_0 = \{ x \rightarrow z^0 \} \land f_1 = \{ x \rightarrow z^0 \}\\
& & s^2 = [ f_0, f_1 ] \Rightarrow f_0 = \{ x \rightarrow z^0 \} \land f_1 = \{ x \rightarrow z^1 \}\\
& & s^3 = [ f_0, f_1 ] \Rightarrow f_0 = \{ x \rightarrow z^0 \} \land f_1 = \{ x \rightarrow z^1, \Re \rightarrow z^1 \}\\
& & s^4 = [ f_0, f_1 ] \Rightarrow f_0 = \{ x \rightarrow z^1 \} \land f_1 = \{ x \rightarrow z^1, \Re \rightarrow z^1 \}\\
& & s^5 = [ f_0 ] \Rightarrow f_0 = \{ x \rightarrow z^1 \}\\
\end{array}
$
\end{center}
}

Clearly, with the same initial states, both algorithms produce the same final state.

\medskip
\noindent{\bf(Induction Hypothesis)} We assume as the induction hypothesis that executing $i$ iterations in both versions with an initial value $z^0$ for $x$ then, if the iterative version obtains a final value $z^n$ for $x$ then the recursive version correctly obtains and stores the same final value $z^n$ for variable $x$ in the top frame.

{\scriptsize
\begin{center}
$
\begin{array}{lc@{~~~~~~~~~~~~~~~~~~~~~}l}
$\underline{Iterative version}$ & & $\underline{Recursive version}$\\
s = [ f_0 ] \Rightarrow f_0 = \{ x \rightarrow z^0 \} & & s = [ f_0 ] \Rightarrow f_0 = \{ x \rightarrow z^0 \}\\
\dots & & \dots\\
s' = [ f_0 ] \Rightarrow f_0 = \{ x \rightarrow z^n \} & & \indent s' = [ f_0 ] \Rightarrow f_0 = \{ x \rightarrow z^n \}\\
\end{array}
$
\end{center}
}

\medskip
\noindent{\bf(Inductive Case)} We now prove that executing $i + 1$ iterations in both versions with an initial value $z^0$ for $x$ then, if the iterative version obtains a final value $z^n$ for $x$ then the recursive version correctly obtains and stores the same final value $z^n$ for variable $x$ in the top frame.\\

The derivation obtained for each version is the following:\\

\noindent \underline{Iterative version}\\
\noindent\\
{\tiny \begin{center}
$\infer{\mathit{env}_0~\vdash~<while~cond~do~S,~s>~\Downarrow~s^2}{
	<cond,~s>~\Rightarrow~true &
	\mathit{env}_0~\vdash~<S,~s>~\Downarrow~s^1&
	\infer{\mathit{env}_0~\vdash~<while~cond~do~S,~s^1>~\Downarrow~s^2}{\mathit{Induction}~\mathit{hypotesis}}
}$
\end{center}}

\noindent \underline{Recursive version}\\
\noindent
{\tiny \begin{center}
$\infer{\bigtriangleup} {
	\infer{\mathit{env}~\vdash~<S; I,~s^1>~\Downarrow~s^4} {
		\mathit{env}~\vdash~<S,~s^1>~\Downarrow~s^2 &
		\infer{\mathit{env}~\vdash~<if~cond~then~t;~return~x,~s^2>~\Downarrow~s^4} {
			\infer{\mathit{env}~\vdash~<if~cond~then~t,~s^2>~\Downarrow~s^3} {\mathit{Induction}~\mathit{hypotesis}} &
			\infer{\mathit{env}~\vdash~<return~x,~s^3>~\Downarrow~s^4} {s^4~=~\mathit{Upd_r}(s^3,~\mathit{Eval}(x, s^3))}
		}
	}
}$
\end{center}}

\noindent
{\tiny \begin{center}
$\infer{\mathit{env}~\vdash~<if~cond~then~t,~s>~\Downarrow~s^6} {
	\infer{\mathit{env}~\vdash~<if~cond~then~t~else~\surd,~s>~\Downarrow~s^6} {
		<cond,~s>~\Rightarrow~true &
		\infer{\mathit{env}~\vdash~<x:= m(a_0, \dots, a_n),~s>~\Downarrow~s^6} {
			(P, I) = env(m) &
			s^1~=~\mathit{AddFrame}(s, P, [a_0, \dots, a_n]) &
			\bigtriangleup &
			s^5~=~\mathit{Upd_{vr}}(s^4, x) &
			s^6~=~\mathit{RemFrame}(s^5)
		}
	}
}$
\end{center}}

\noindent
Because both algorithms have the same initial value $z^0$ for $x$ then the states during the iteration are modified as follows (the * state is obtained by the induction hypothesis):\\

{\scriptsize
\begin{center}
$
\begin{array}{lc@{~~~~~~~~~~~~~~~~~~~~~}l}
$\underline{Iterative version}$ & & $\underline{Recursive version}$\\
s = [ f_0 ] \Rightarrow f_0 = \{ x \rightarrow z^0 \} & & s = [ f_0 ] \Rightarrow f_0 = \{ x \rightarrow z^0 \}\\
s^1 = [ f_0 ] \Rightarrow f_0 = \{ x \rightarrow z^1 \} & & s^1 = [ f_0, f_1 ] \Rightarrow f_0 = \{ x \rightarrow z^0 \} \land f_1 = \{ x \rightarrow z^0 \}\\
s^2 = [ f_0 ] \Rightarrow f_0 = \{ x \rightarrow z^n \} * & & s^2 = [ f_0, f_1 ] \Rightarrow f_0 = \{ x \rightarrow z^0 \} \land f_1 = \{ x \rightarrow z^1 \}\\
 & & s^3 = [ f_0, f_1 ] \Rightarrow f_0 = \{ x \rightarrow z^0 \} \land f_1 = \{ x \rightarrow z^n \} *\\
 & & s^4 = [ f_0, f_1 ] \Rightarrow f_0 = \{ x \rightarrow z^0 \} \land f_1 = \{ x \rightarrow z^n, \Re \rightarrow z^n \}\\
 & & s^5 = [ f_0, f_1 ] \Rightarrow f_0 = \{ x \rightarrow z^n \} \land f_1 = \{ x \rightarrow z^n, \Re \rightarrow z^n \}\\
 & & s^6 = [ f_0 ] \Rightarrow f_0 = \{ x \rightarrow z^n \}\\
\end{array}
$
\end{center}
}

\noindent

\noindent Hence, Algorithm~\ref{alg_code2} and Algorithm~\ref{alg_code1} obtain the same final state, and thus, they are semantically equivalent.
\end{proof}

\section{Conclusions}
\label{sec_concl}
 
Transforming loops to recursion is useful in many situations such as, e.g., debugging, verification or memory hierarchies optimization. It is therefore surprising that there did not exist an automatic transformation from loops to recursion, but it is even more surprising that no public report exists that describes how to implement this transformation. 

In this article the transformation of each kind of Java loop 
has been described independently with a specific treatment for it that is illustrated with an example of use.
Moreover, the transformation has been described in such a way that it can be easily adapted to any other language where recursion and iteration exist. 

\bibliography{\biblio{biblio}}
\bibliographystyle{abbrv}

\end{document}